\documentclass{article}

\usepackage{fullpage, amsmath, amssymb, amsthm}
\usepackage[usenames]{color}

\let\Pr\relax
\DeclareMathOperator*{\Pr}{\mathbb{P}}

\newcommand{\R}{\mathbb{R}}

\newcommand{\Fam}{\mathcal{F}}

\newcommand{\eps}{\varepsilon}

\newcommand{\tr}{\textrm{tr}}

\newtheorem{theorem}{Theorem}
\newtheorem{corollary}{Corollary}
\newtheorem{lemma}{Lemma}

\newtheorem{remark}{Remark}

\begin{document}

\date{}

\author{Kasper Green Larsen\thanks{Aarhus University. \texttt{larsen@cs.au.dk}. Supported by Center for Massive Data Algorithmics, a Center of the Danish National Research Foundation, grant DNRF84.}
  \and Jelani Nelson\thanks{Harvard University. \texttt{minilek@seas.harvard.edu}. Supported by NSF
  CAREER award CCF-1350670, NSF grant IIS-1447471, ONR grant N00014-14-1-0632, and a Google Faculty Research Award.}}

\title{The Johnson-Lindenstrauss lemma is optimal for linear dimensionality reduction}

\maketitle

\begin{abstract}
For any $n>1$ and $0<\eps<1/2$, we show the existence of an $n^{O(1)}$-point subset $X$ of $\R^n$ such that any linear map from $(X,\ell_2)$ to $\ell_2^m$ with distortion at most $1+\eps$ must have $m = \Omega(\min\{n, \eps^{-2}\log n\})$. Our lower bound matches the upper bounds provided by the identity matrix and the Johnson-Lindenstrauss lemma \cite{JL84}, improving the previous lower bound of Alon \cite{Alon03} by a $\log(1/\eps)$ factor.
\end{abstract}

\section{Introduction}
The Johnson-Lindenstrauss lemma \cite{JL84} states the following.

\begin{theorem}[JL lemma {\cite[Lemma 1]{JL84}}]
For any $N$-point subset $X$ of Euclidean space and any $0<\eps<1/2$, there exists a map $f:X\rightarrow\ell_2^m$ with $m = O(\eps^{-2}\log N)$ such that
\begin{equation}
\forall x,y\in X,\ (1-\eps)\|x-y\|_2^2 \le \|f(x)-f(y)\|_2^2 \le (1+\eps)\|x-y\|_2^2 . \label{eqn:jl-guarantee}
\end{equation}
\end{theorem}

We henceforth refer to $f$ satisfying \eqref{eqn:jl-guarantee} as {\em having the $\eps$-JL guarantee for $X$} (often we drop mention of $\eps$ when understood from context). The JL lemma has found applications in computer science, signal processing (e.g.\ compressed sensing), statistics, and mathematics. The main idea in algorithmic applications is that one can transform a high-dimensional problem into a low-dimensional one such that an optimal solution to the lower dimensional problem can be lifted to a nearly optimal solution to the original problem. Due to the decreased dimension, the lower dimensional problem requires fewer resources (time, memory, etc.) to solve. We refer the reader to \cite{Indyk01,Vempala04,Matousek08} for a list of further applications.

All known proofs of the JL lemma with target dimension as stated above in fact provide such a map $f$ which is {\em linear}. This linearity property is important in several applications. For example in the turnstile model of streaming \cite{Muthukrishnan05}, a vector $x\in\R^n$ receives a stream of coordinate-wise updates each of the form $x_i\leftarrow x_i + \Delta$, where $\Delta\in\R$. The goal is to process $x$ using some $m\ll n$ memory. Thus if one wants to perform dimensionality reduction in a stream, which occurs for example in streaming linear algebra applications \cite{ClarksonW09}, this can be achieved with linear $f$ since $f(x + \Delta\cdot e_i) = f(x) + \Delta\cdot f(e_i)$. In compressed sensing, another application where linearity of $f$ is inherent, one wishes to (approximately) recover (approximately) sparse signals using few linear measurements \cite{Donoho04,CT05}. The map $f$ sending a signal to the vector containing some fixed set of linear measurements of it is known to allow for good signal recovery as long as $f$ satisfies the JL guarantee for the set of all $k$-sparse vectors \cite{CT05}.

Given the widespread use of dimensionality reduction across several domains, it is a natural and often-asked question  whether the JL lemma is tight: does there exist some $X$ of size $N$ such that any such map $f$ must have $m = \Omega(\eps^{-2}\log N)$? The paper \cite{JL84} introducing the JL lemma provided the first lower bound of $m = \Omega(\log N)$ when $\eps$ is smaller than some constant. This was improved by Alon \cite{Alon03}, who showed that if $X = \{0,e_1,\ldots,e_n\}\subset\R^n$ is the simplex (thus $N = n+1$) and $0 < \eps < 1/2$, then any JL map $f$ must embed into dimension $m = \Omega(\min\{n, \eps^{-2}\log n/\log(1/\eps)\})$. Note the first term in the $\min$ is achieved by the identity map. Furthermore, the $\log(1/\eps)$ term cannot be removed for this particular $X$ since one can use Reed-Solomon codes to obtain embeddings with $m = O(1/\eps^2)$ (superior to the JL lemma) once $\eps \le n^{-\Omega(1)}$ \cite{Alon03} (see \cite{NNW14} for details). Specifically, for this $X$ it is possible to achieve $m = O(\eps^{-2}\min\{\log N, ((\log N)/\log(1/\eps))^2\})$. Note also for this choice of $X$ we can assume that any $f$ is in fact linear. This is because first we can assume $f(0) = 0$ by translation. Then we can form a matrix $A\in\R^{m\times n}$ such that the $i$th column of $A$ is $f(e_i)$. Then trivially $A e_i = f(e_i)$ and $A 0 = 0 = f(0)$. 

The fact that the JL lemma is not optimal for the simplex for small $\eps$ begs the question: is the JL lemma suboptimal for all point sets? This is a major open question in the area of dimensionality reduction, and it has been open since the paper of Johnson and Lindenstrauss 30 years ago.

\paragraph{Our Main Contribution:} For any $n>1$ and $0<\eps<1/2$, there is an $n^{O(1)}$-point subset $X$ of $\R^n$ such that any embedding $f$ providing the JL guarantee, and where {\em $f$ is linear}, must have $m = \Omega(\min\{n, \eps^{-2}\log n\})$. In other words, the JL lemma is optimal in the case where $f$ must be linear.

\bigskip

Our lower bound is optimal: the identity map achieves the first term in the $\min$, and the JL lemma provides the second. Our lower bound is only against linear embeddings, but as stated before: (1) all known proofs of the JL lemma give linear $f$, and (2) for several applications it is important that $f$ be linear.

Before our work there were two possibilities for dimensionality reduction in $\ell_2$: (i) the target dimension $m$ could be reduced for all point sets $X$, at least for small $\eps$ as with the simplex using Reed-Solomon codes, or (ii) there is a higher lower bound for some other point set $X$ which is harder than the simplex. Evidence existed to support both possibilities. On the one hand the simplex was the hardest case in many respects: it gave the highest lower bound known on $m$ \cite{Alon03}, and it also was a hardest case for the data-dependent upper bound on $m$ of Gordon \cite{Gordon88} (involving the gaussian mean width of the normalized difference vectors $X-X$; we will not delve deeper into this topic here). Meanwhile for (ii), random linear maps were the only JL construction we knew for arbitrary $X$, and such an approach with random maps is known to require $m = \Omega(\min\{n, \eps^{-2}\log N\})$ \cite{JayramW13,KaneMN11} (see Remark~\ref{rmk:djl} below for details). 

Thus given the previous state of our knowledge, it was not clear which was more likely between worlds (i) and (ii). Our lower bound gives more support to (ii), since we not only rule out further improvements to JL using random linear maps, but rule out improvements using {\em any} linear map. Furthermore all known methods for efficient dimensionality reduction in $\ell_2$ are via linear maps, and thus circumventing our lower bound would require a fundamentally new approach to the problem. We also discuss in Section~\ref{sec:discussion} what would be needed to push our lower bound to apply to non-linear maps.

\begin{remark}\label{rmk:djl}
\textup{
It is worth noting that the JL lemma is different from the {\em distributional} JL (DJL) lemma that often appears in the literature, sometimes with the same name (though the lemmas are different!). In the DJL problem, one is given an integer $n>1$ and $0<\eps,\delta<1/2$, and the goal is to provide a distribution $\mathcal{F}$ over maps $f:\ell_2^n\rightarrow \ell_2^m$ with $m$ as small as possible such that for any fixed $x\in\R^n$
$$
\Pr_{f\leftarrow\mathcal{F}}(\|f(x)\|_2 \notin [(1-\eps)\|x\|_2,(1+\eps)\|x\|_2]) < \delta .
$$
The existence of such $\mathcal{F}$ with small $m$ implies the JL lemma by taking $\delta < 1/\binom{N}{2}$. Then for any $z\in X-X$, a random $f\leftarrow\mathcal{F}$ fails to preserve the norm of $z$ with probability $\delta$. Thus the probability that there exists $z\in X-X$ which $f$ fails to preserve the norm of is at most $\delta\cdot \binom{N}{2} < 1$, by a union bound. In other words, a random map provides the desired JL guarantee with high probability (and in fact this map is chosen completely obliviously of the input vectors).
}

\textup{
The optimal $m$ for the DJL lemma when using linear maps is understood. The original paper \cite{JL84} provided a linear solution to the DJL problem with $m = O(\min\{n,\eps^{-2}\log(1/\delta)\})$, and this was later shown to be optimal for the full range of $\eps,\delta\in(0,1/2)$ \cite{JayramW13,KaneMN11}. Thus when $\delta$ is set as above, one obtains the $m = O(\eps^{-2}\log N)$ guarantee of the JL lemma. However, this does not imply that the JL lemma is tight. Indeed, it is sometimes possible to obtain smaller $m$ by avoiding the DJL lemma, such as the Reed-Solomon based embedding construction for the simplex mentioned above (which involves zero randomness).
}

\textup{
It is also worth remarking that DJL is desirable for one-pass streaming algorithms, since no properties of $X$ are known when the map $f$ is chosen at the beginning of the stream, and thus the DJL lower bounds of \cite{JayramW13,KaneMN11} are relevant in this scenario. However when allowed two passes or more, one could imagine estimating various properties of $X$ in the first pass(es) then choosing some linear $f$ more efficiently based on these properties to perform dimensionality reduction in the last pass. The lower bound of our main theorem shows that the target dimension could not be reduced by such an approach.
}
\end{remark}

\subsection{Proof overview}

For any $n>1$ and $\eps\in(0,1/2)$, we prove the existence of $X\subset\R^n$, $|X| = N = O(n^3)$, s.t.\ if for $A\in\R^{m\times n}$
\begin{equation}
(1-\eps)\|x\|_2^2 \le \|Ax\|_2^2 \le (1+\eps)\|x\|_2^2 \text{ for all } x\in X , \label{eqn:jl-guarantee2}
\end{equation}
then $m = \Omega(\eps^{-2}\log n) = \Omega(\eps^{-2}\log N)$. Providing the JL guarantee on $X\cup\{0\}$ implies satisfying \eqref{eqn:jl-guarantee2}, and therefore also requires $m = \Omega(\eps^{-2}\log N)$. We show such $X$ exists via the probabilistic method, by letting  $X$ be the union of all $n$ standard basis vectors together with several independent gaussian vectors. Gaussian vectors were also the hard case in the DJL lower bound proof of \cite{KaneMN11}, though the details were different.

We now give the idea of the lower bound proof to achieve \eqref{eqn:jl-guarantee2}. First, we include in $X$ the vectors $e_1,\ldots,e_n$. Then if $A\in\R^{m\times n}$ for $m\le n$ satisfies \eqref{eqn:jl-guarantee2}, this forces every column of $A$ to have roughly unit norm. Then by standard results in covering and packing (see Eqn.\ (5.7) of \cite{Pisier89}), there exists some family of matrices $\mathcal{F}\subset\cup_{t=1}^n \R^{t\times n}$, $|\mathcal{F}| = e^{O(n^2\log n)}$, such that
\begin{equation}
\inf_{\hat{A}\in\mathcal{F}\cap\R^{m\times n}} \| A - \hat{A} \|_F \le \frac{1}{n^C}\label{eqn:packing}
\end{equation}
for $C>0$ a constant as large as we like, where $\|\cdot\|_F$ denotes Frobenius norm. Also, by a theorem of Lata{\l}a \cite{Latala99}, for any $\hat{A}\in\mathcal{F}$ and for a random gaussian vector $g$,
\begin{equation}
\Pr_g(|\|\hat{A} g\|_2^2 - \tr(\hat{A}^T \hat{A})| \ge \Omega(\sqrt{\log(1/\delta)} \cdot \|\hat{A}^T \hat{A}\|_F)) \ge \delta \label{eqn:rev-hw}
\end{equation}
for any $0<\delta<1/2$, where $\tr(\cdot)$ is trace. This is a (weaker version of the) statement that for gaussians, the Hanson-Wright inequality \cite{HW71} not only provides an upper bound on the tail of degree-two gaussian chaos, but also is a lower bound. (The strong form of the previous sentence, without the parenthetical qualifier, was proven in \cite{Latala99}, but we do not need this stronger form for our proof -- essentially the difference is that in stronger form, \eqref{eqn:rev-hw} is replaced with a stronger inequality also involving the operator norm $\|\hat{A}^T \hat{A}\|$.)

It also follows by standard gaussian concentration that a random gaussian vector $g$ satisfies
\begin{equation}
\Pr_g(|\|g\|_2^2 - n| > C\sqrt{n\log(1/\delta)}) < \delta/2 \label{eqn:gauss-conc}
\end{equation}

Thus by a union bound, the events of \eqref{eqn:rev-hw}, \eqref{eqn:gauss-conc} happen simultaneously with probability $\Omega(\delta)$. Thus if we take $N$ random gaussian vectors, the probability that the events of \eqref{eqn:rev-hw}, \eqref{eqn:gauss-conc} never happen simultaneously for any of the $N$ gaussians is at most $(1 - \Omega(\delta))^N = e^{-\Omega(\delta N)}$. By picking $N$ sufficiently large and $\delta = 1/\mathrm{poly}(n)$, a union bound over $\mathcal{F}$ shows that for every $\hat{A}\in\mathcal{F}$, one of the $N$ gaussians satisfies the events of \eqref{eqn:rev-hw} and \eqref{eqn:gauss-conc} simultaneously. Specifically, there exist $N = O(n^3)$ vectors $\{v_1,\ldots,v_N\} = V\subset\R^n$ such that
\begin{itemize}
\item Every $v\in V$ has $\|v\|_2^2 = n \pm O(\sqrt{n\lg n})$
\item For any $\hat{A}\in\mathcal{F}$ there exists some $v\in V$ such that $|\|\hat{A} v\|_2^2 - \tr(\hat{A}^T\hat{A})| = \Omega(\sqrt{\lg n}\cdot \|\hat{A}\|_F)$.
\end{itemize}
The final definition of $X$ is $\{e_1,\ldots,e_n\}\cup V$. Then, using \eqref{eqn:jl-guarantee2} and \eqref{eqn:packing}, we show that the second bullet implies 
\begin{equation}
\tr(\hat{A}^T\hat{A}) = n \pm O(\eps n),\text{ and }|\|A v\|_2^2 - n| = \Omega(\sqrt{\ln n}\cdot \|\hat{A}^T\hat{A}\|_F) - O(\eps n) . \label{eqn:lb}
\end{equation}
where $\pm B$ represents a value in $[-B,B]$. But then by the triangle inequality, the first bullet above, and \eqref{eqn:jl-guarantee2},
\begin{equation}
\left|\|A v\|_2^2 - n\right| \le \left|\|A v\|_2^2 - \|v\|_2^2\right| + \left|\|v|\|_2^2 - n\right| = O(\eps n + \sqrt{n\lg n}) . \label{eqn:ub}
\end{equation}

Combining \eqref{eqn:lb} and \eqref{eqn:ub} implies
$$
\tr(\hat{A}^T\hat{A}) = \sum_{i=1}^n \hat{\lambda}_i \ge (1-O(\eps))n,\text{ and }\|\hat{A}^T\hat{A}\|_F^2 = \sum_{i=1}^n \hat{\lambda}_i^2 = O\left(\frac{\eps^2 n^2}{\log n} + n\right)
$$
where $(\hat{\lambda}_i)$ are the eigenvalues of $\hat{A}^T \hat{A}$. With bounds on $\sum_i \hat{\lambda_i}$ and $\sum_i \hat{\lambda_i}^2$ in hand, a lower bound on $\mathrm{rank}(\hat{A}^T\hat{A}) \le m$ follows by Cauchy-Schwarz (this last step is also common to the proof of \cite{Alon03}).

\begin{remark}
\textup{
It is not crucial in our proof that $N$ be proportional to $n^3$. Our techniques straightforwardly extend to show that $N$ can be any value which is $\Omega(n^{2+\gamma})$ for any constant $\gamma>0$.
}
\end{remark}

\section{Preliminaries}
Henceforth a {\em standard gaussian} random variable $g\in\R$ is a gaussian with mean $0$ and variance $1$. If we say $g\in\R^n$ is standard gaussian, then we mean that $g$ is a multivariate gaussian with identity covariance matrix (i.e.\ its entries are independent standard gaussian). Also, the notation $\pm B$ denotes a value in $[-B,B]$. For a real matrix $A=(a_{i,j})$, $\|A\|$ is the $\ell_2\rightarrow\ell_2$ operator norm, and $\|A\|_F = (\sum_{i,j} a_{i,j}^2)^{1/2}$ is Frobenius norm.

In our proof we depend on some previous work. The first theorem is due to Lata{\l}a \cite{Latala99} and says that, for gaussians, the Hanson-Wright inequality is not only an upper bound but also a lower bound.

\begin{theorem}[{\cite[Corollary 2]{Latala99}}]\label{thm:latala}
There exists universal $c>0$ such that for $g\in\R^n$ standard gaussian and $A = (a_{i,j})$ an $n\times n$ real symmetric matrix with zero diagonal,
$$
\forall t\ge 1,\ \Pr_g\left(|g^T A g| > c(\sqrt{t}\cdot \|A\|_F + t\cdot \|A\|)\right) \ge \min\{c, e^{-t}\}
$$
\end{theorem}

Theorem~\ref{thm:latala} implies the following corollary.

\begin{corollary}\label{cor:trace}
Let $g,A$ be as in Theorem~\ref{thm:latala}, but where $A$ is no longer restricted to have zero diagonal. Then
$$
\forall t\ge 1,\ \Pr_g\left(|g^T A g - \tr(A)| > c(\sqrt{t}\cdot \|A\|_F + t\cdot \|A\|)\right) \ge \min\{c, e^{-t}\}
$$
\end{corollary}
\begin{proof}
Let $N$ be a positive integer. Define $\tilde{g} = (\tilde{g}_{1,1},\tilde{g}_{1,2},\ldots,\tilde{g}_{1,N},\ldots,\tilde{g}_{n,1},\tilde{g}_{n,2},\ldots,\tilde{g}_{n,N})$ a standard gaussian vector. Then $g_i$ is equal in distribution to $N^{-1/2}\sum_{j=1}^N \tilde{g}_{i,j}$. Define $\tilde{A}_N$ as the $nN\times nN$ matrix formed by converting each entry $a_{i,j}$ of $A$ into an $N\times N$ block with each entry being $a_{i,j}/N$. Then
$$
g^T A g - \tr(A) = \sum_{i=1}^n\sum_{j=1}^n a_{i,j} g_i g_j - \tr(A) \mathbin{\stackrel{d}{=}} \sum_{i=1}^n\sum_{j=1}^n \sum_{r=1}^N \sum_{s=1}^N \frac{a_{i,j}}N\tilde{g}_{i,r}\tilde{g}_{j,s} - \tr(A) \mathbin{\stackrel{\rm def}{=}} \tilde{g}^T \tilde{A}_N \tilde{g} - \tr(\tilde{A}_N)
$$ 
where $\mathbin{\stackrel{d}{=}}$ denotes equality in distribution (note $\tr(A) = \tr(\tilde{A}_N)$). By the weak law of large numbers
\begin{equation}
\forall\lambda > 0,\ \lim_{N\rightarrow\infty} \Pr\left(|\tilde{g}^T \tilde{A}_N \tilde{g} - \tr(\tilde{A}_N)| > \lambda \right) = \lim_{N\rightarrow\infty} \Pr\left(|\tilde{g}^T (\tilde{A}_N - \tilde{D}_N)  \tilde{g}| > \lambda \right) \label{eqn:no-diag}
\end{equation}
where $\tilde{D}_N$ is diagonal containing the diagonal elements of $\tilde{A}_N$. Note $\|\tilde{A}_N\| = \|A\|$. This follows since if we have the singular value decomposition $A = \sum_i \sigma_i u_i v_i^T$ (where the $\{u_i\}$ and $\{v_i\}$ are each orthonormal, $\sigma_i>0$, and $\|A\|$ is the largest of the $\sigma_i$), then $\tilde{A}_N = \sum_i \sigma_i u^{(N)}_i (v^{(N)}_i)^T$ where $u^{(N)}_i$ is equal to $u_i$ but where every coordinate is replicated $N$ times and divided by $\sqrt{N}$.  This implies $|\|\tilde{A}_N - \tilde{D}_N\| - \|A\|| \le \|\tilde{D}_N\| = \max_i |a_{i,i}| / N = o_N(1)$ by the triangle inequality. Therefore $\lim_{N\rightarrow\infty}\|\tilde{A}_N - \tilde{D}_N\| = \|A\|$. Also $\lim_{N\rightarrow\infty}\|\tilde{A}_N - \tilde{D}_N\|_F = \|A\|_F$. Our corollary follows by applying Theorem~\ref{thm:latala} to the right side of \eqref{eqn:no-diag}.
\end{proof}

The next lemma follows from gaussian concentration of Lipschitz functions \cite[Corollary 2.3]{Pi86}. It also follows directly from the Hanson-Wright inequality \cite{HW71}.

\begin{lemma}\label{lem:conc-norm}
For some universal $c>0$ and $g\in\R^n$ a standard gaussian, $\forall t > 0\ \Pr(|\|g\|_2^2 - n| > c\sqrt{nt}) < e^{-t}$.
\end{lemma}

The following corollary summarizes the above in a form that will be useful later.

\begin{corollary}
\label{cor:many}
For $A \in \R^{d \times n}$ let $\lambda_1 \geq \cdots \geq \lambda_n \geq 0$ be the eigenvalues of $A^TA$. Let $g^{(1)},\ldots,g^{(N)}\in\R^n$ be independent standard gaussian vectors. For some universal constants $c_1,c_2,\delta_0>0$ and any $0 < \delta < \delta_0$
\begin{equation}
\Pr\left( \not\exists j\in[N] : \Bigg\{\left|\|Ag^{(j)}\|_2^2 - \sum_{i=1}^n \lambda_i\right| \ge c_1\sqrt{\ln(1/\delta)}\left( \sum_{i=1}^n \lambda_i^2\right)^{1/2}\Bigg\} \wedge \Bigg\{|\|g^{(j)}\|_2^2 - n| \le
c_2\sqrt{n\ln(1/\delta)}\Bigg\}\right) \leq e^{-N \delta}.\label{eqn:goodj}
\end{equation}
\end{corollary}
\begin{proof}
We will show that for any fixed $j\in [N]$ it holds that
\begin{equation}\label{eqn:bothgood}
\Pr\left(\Bigg\{\left|\|Ag^{(j)}\|_2^2 - \sum_{i=1}^n \lambda_i\right|\geq c_1\sqrt{\ln(1/\delta)}\left( \sum_{i=1}^n \lambda_i^2\right)^{1/2}\Bigg\} \wedge \Bigg\{\|g^{(j)}\|_2^2 \leq n + c_2\sqrt{n\ln(1/\delta)}\Bigg\}\right) > \delta 
\end{equation}
Then, since the $g_j$ are independent, the left side of \eqref{eqn:goodj} is at most $(1-\delta)^N \le e^{-\delta N}$.

Now we must show \eqref{eqn:bothgood}. It suffices to show that
\begin{equation}
\Pr\left(|\|g^{(j)}\|_2^2 - n| \le c_2\sqrt{n\ln(1/\delta)}\right) > 1 - \delta/2 \label{eqn:lipschitz-app}
\end{equation}
and
\begin{equation}
\Pr\left(\left|\|Ag^{(j)}\|_2^2 - \sum_{i=1}^n \lambda_i\right| \geq  c_1\sqrt{\ln(1/\delta)}\left( \sum_{i=1}^n \lambda_i^2\right)^{1/2}\right) > \delta/2 \label{eqn:latala-app}
\end{equation}
since \eqref{eqn:bothgood} would then follow from a union bound. Eqn.~\eqref{eqn:lipschitz-app} follows immediately from Lemma~\ref{lem:conc-norm} for $c_2$ chosen sufficiently large. For Eqn.~\eqref{eqn:latala-app}, note $\|Ag^{(j)}\|_2^2 = g^T A^T A g$. Then $\sum_i \lambda_i = \tr(A^TA)$ and $(\sum_i \lambda_i^2)^{1/2} = \|A^TA\|_F$. Then \eqref{eqn:latala-app} frollows from Corollary~\ref{cor:trace} for $\delta$ smaller than some sufficiently small constant $\delta_0$.
\end{proof}

We also need a standard estimate on entropy numbers (covering the unit $\ell_\infty^{mn}$ ball by $\ell_2^{mn}$ balls).

\begin{lemma}\label{lem:fam}
For any parameter $0 < \alpha < 1$,
there exists a family $\Fam_{\alpha} \subseteq \bigcup_{m=1}^n \R^{m \times n}$ of matrices
with the following two properties:
\begin{enumerate}
\item For any
matrix $A \in \bigcup_{m=1}^n \R^{m \times n}$ having all entries
bounded in absolute value by $2$, there is a matrix
$\hat{A} \in \Fam_{\alpha} $ such that $A$ and $\hat{A}$
have the same number of rows and $B = A-\hat{A}$ satisfies
$\tr(B^TB) \leq \alpha/100$.
\item $|\Fam_{\alpha}| = e^{O(n^2 \ln(n/\alpha))}$.
\end{enumerate}
\end{lemma}
\begin{proof}
We construct $\Fam_{\alpha}$ as follows: For each integer $1 \leq m
\leq n$, add all $m \times n$ matrices having entries of the form $i
\frac{\sqrt{\alpha}}{10n}$ for integers $i \in \{-20n/\sqrt{\alpha},\dots, 20n/\sqrt{\alpha}\}$. Then
for any matrix $A \in \bigcup_{m=1}^n \R^{m \times n}$ there is a matrix $\hat{A} \in
\Fam_{\alpha}$ such that $A$ and $\hat{A}$ have the same number of
rows and every entry of $B=A-\hat{A}$ is bounded in absolute
value by $\frac{\sqrt{\alpha}}{10n}$. This means that every diagonal
entry of $B^TB$ is bounded by $n\alpha/(100n^2)$ and thus $\tr(B^TB) \leq \alpha/100$. The
size of $\Fam_{\alpha}$ is bounded by
$n(40n/\sqrt{\alpha})^{n^2} = e^{O(n^2 \ln(n/\alpha))}$.
\end{proof}

\section{Proof of main theorem}

\begin{lemma}
\label{lem:largescaling}
Let $\Fam_{\alpha}$ be as in Lemma~\ref{lem:fam} with $1/\mathop{poly}(n)\le \alpha < 1$.  Then there exists a set of $N=O(n^3)$ vectors $v_1,\dots,v_N$ in $\R^n$ such that for every matrix $A \in \Fam_{\alpha}$, there is an index $j \in [N]$ such that 
\begin{enumerate}
\item[(i)] $|\|Av_j\|_2^2 - \sum_i \lambda_i| = \Omega\left( \sqrt{\ln n \sum_i \lambda_i^2}\right).$
\item[(ii)] $|\|v_j\|_2^2-  n| = O(\sqrt{n \ln n})$.
\end{enumerate}
\end{lemma}

\begin{proof}
Let $g^{(1)}, \dots, g^{(N)}\in\R^n$ be independent standard gaussian. Let $A \in \Fam_{\alpha}$ and apply Corollary~\ref{cor:many} with $\delta = n^{-1/4} = N^{-1/12}$. With probability $1-e^{-\Omega(n^{3-1/4})}$, one of the $g^{(j)}$ for $j\in[N]$ satisfies (i) and (ii) for $A$. Since $|\Fam_{\alpha}|= e^{O(n^2 \ln (n/\alpha))}$, the claim follows by a union bound over all matrices in $\Fam_\alpha$.
\end{proof}

\begin{theorem}
For any $0 < \eps < 1/2$, there exists a set $V\subset\R^n$, $|V|=N=n^3+n$,
such that if $A$ is a matrix in $\R^{m \times n}$ satisfying
$\|Av_i\|_2^2 \in (1 \pm \eps)\|v_i\|_2^2$ for all $v_i \in V$, then
$m = \Omega(\min\{n, \eps^{-2} \lg n\})$.
\end{theorem}

\begin{proof}
We can assume $\eps>1/\sqrt{n}$ since otherwise an $m = \Omega(n)$ lower bound already follows from \cite{Alon03}.
To construct $V$, we first invoke Lemma~\ref{lem:largescaling} with $\alpha=\eps^2/n^2$ to
find $n^3$ vectors $w_1,\dots,w_{n^3}$ such that for all matrices $\tilde{A} \in
\Fam_{\eps^2/n^2}$, there exists an index $j \in [n^3]$ for which:
\begin{enumerate}
\item $|\|\tilde{A}w_j\|_2^2 - \sum_i \tilde{\lambda}_i| \ge \Omega\left( \sqrt{(\ln n) \sum_i \tilde{\lambda}_i^2}\right).$
\item $|\|w_j\|_2^2 - n| = O(\sqrt{n \ln n})$.
\end{enumerate}
where $\tilde{\lambda}_1 \geq \cdots \geq \tilde{\lambda}_n \geq 0$ denote the eigenvalues
of $\tilde{A}^T\tilde{A}$. We let $V = \{e_1,\dots,e_n, w_1,\dots,w_{n^3}\}$ and claim
this set of $N=n^3+n$ vectors satisfies the theorem. Here $e_i$ denotes the
$i$'th standard unit vector.

To prove this, let $A \in \R^{m \times n}$ be a matrix with $m \leq n$
satisfying
$\|Av\|_2^2 \in (1 \pm \eps)\|v\|_2^2$ for all $v \in V$. Now observe
that since $e_1,\dots,e_n \in V$, $A$ satisfies $\|Ae_i\|_2^2 \in (1 \pm \eps)\|e_i\|_2^2 = (1
\pm \eps)$ for all $e_i$. Hence all entries $a_{i,j}$ of $A$ must 
have $a_{i,j}^2 \le (1+\eps) < 2$ (and in fact, all columns of $A$
have $\ell_2$ norm at most $\sqrt{2}$). This implies
that there is an $m \times n$ matrix $\hat{A} \in \Fam_{\eps^2/n^2}$ such that
$B = A-\hat{A} = (b_{i,j})$ satisfies $\tr(B^TB) \leq
\eps^2/(100n^2)$. Since $\tr(B^TB) = \|B\|_F^2$, this also implies
$\|B\|_F \le \eps/(10 n)$. Then by Cauchy-Schwarz,
\begin{align*}
\sum_{i=1}^n \hat{\lambda}_i &= \tr(\hat{A}^T \hat{A})\\
{}&= \tr((A-B)^T(A-B))\\
{}&=\tr(A^T A) + \tr(B^T B) - \tr(A^T B) - \tr(B^T A)\\
{}&=\sum_{i=1}^n \|Ae_i\|_2^2 + \tr(B^T B) - \tr(A^T B) - \tr(B^T A)\\
{}&= n \pm (O(\eps n) + 2n\cdot \max_j (\sum_i b_{i,j}^2)^{1/2}\cdot \max_k (\sum_i a_{i,k}^2)^{1/2})\\
{}& = n \pm (O(\eps n)+ 2n\cdot \|B\|_F\cdot \sqrt{2})\\
{}& =  n \pm O(\eps n).
\end{align*}

Thus from our choice of $V$ there exists a vector $v^* \in V$
such that
\begin{enumerate}
\item[(i)] $|\|\hat{A}v^*\|_2^2 -n| \ge \Omega\left(\sqrt{(\ln n) \sum_i \hat{\lambda}_i^2}\right) - O(\eps n).$
\item[(ii)] $|\|v^*\|_2^2 - n| = O(\sqrt{n \ln n})$.
\end{enumerate}
Note $\|B\|^2 \le \|B\|_F^2 = \tr(B^T B) \le \eps^2/(100 n^2)$ and $\|\hat{A}\|^2 \le \|\hat{A}\|_F^2 \le (\|A\|_F + \|B\|_F)^2 = O(n^2)$. Then by (i)
\allowdisplaybreaks
\begin{enumerate}
\item[(iii)]
\begin{align*}
|\|Av^*\|_2^2 - n| &= |\|\hat{A}v^*\|_2^2 + \|Bv^*\|_2^2 + 2\langle\hat{A}v^*,Bv^*\rangle -n|\\
&{}\geq \Omega\left(\sqrt{(\ln n) \sum_i \hat{\lambda}_i^2}\right) - \|Bv^*\|_2^2 - 2|\langle\hat{A}v^*,Bv^*\rangle|-O(\eps n)\\
&{}\geq \Omega\left(\sqrt{(\ln n) \sum_i \hat{\lambda}_i^2}\right) - \|B\|^2\cdot\|v^*\|_2^2 - 2\|B\|\cdot\|A\|\cdot\|v^*\|_2^2-O(\eps n)\\
&{}=\Omega\left(\sqrt{(\ln n) \sum_i \hat{\lambda}_i^2}\right)-O(\eps n).
\end{align*}
\end{enumerate}
We assumed $|\|Av^*\|_2^2 - \|v^*\|_2^2| = O(\eps \|v^*\|_2^2) =  O(\eps n)$. Therefore by (ii),
\begin{align*}
\left|\|Av^*\|_2^2 -n\right| &\le \left|\|Av^*\|_2^2 - \|v^*\|_2^2\right| + \left|\|v^*\|_2^2 - n\right|\\
{}&= O(\eps n + \sqrt{n\ln n}) ,
\end{align*}
which when combined with (iii) implies
$$
\sum_{i=1}^n \hat{\lambda}_i^2 = O\left(\frac{\eps^2n^2}{\ln n} + n\right).
$$

To complete the proof, by Cauchy-Schwarz since exactly $\mathrm{rank}(\hat{A}^T\hat{A})$ of the $\hat{\lambda}_i$ are non-zero,
$$
\frac{n^2}2 \le \left(\sum_{i=1}^n \hat{\lambda}_i\right)^2 \le \mathrm{rank}(\hat{A}^T\hat{A}) \cdot\left(\sum_{i=1}^n \hat{\lambda_i}^2\right) \le m\cdot O\left(\frac{\eps^2n^2}{\ln n} + n\right)
$$
Rearranging gives $m = \Omega(\min\{n, \eps^{-2}\ln n\}) = \Omega(\min\{n, \eps^{-2}\ln N\})$ as desired.

\end{proof}

\section{Discussion}\label{sec:discussion}

One obvious future direction is to obtain an $m = \Omega(\min\{n, \eps^{-2}\log N\})$ lower bound that also applies to non-linear maps. Our hard set $X$ contains $N = O(n^3)$ points in $\R^n$ (though as remarked earlier, our techniques easily imply $N = O(n^{2+\gamma})$ points suffice). Any embedding $f$ could be assumed linear without loss of generality if the elements of $X$ were linearly independent, but clearly this cannot happen if $N > n$ (as is the case for our $X$). Thus a first step toward a lower bound against non-linear embeddings is to obtain a hard $X$ with $N=|X|$ as small as possible. One step in this direction could be the following. Observe that our lower bound only uses that $\|f(x)\|_2 \approx \|x\|_2$ for each $x\in X$, whereas the full JL lemma requires that all distance vectors $X-X$ have their norms preserved. Thus one could hope to exploit this fact and take $|X| = \Theta(n^{1+\gamma})$, say, since then $X - X$ would still have the $\Theta(n^{2+\gamma})$ points needed to carry out the union bound of Lemma~\ref{lem:largescaling}. The problem is that these $\Theta(n^{2+\gamma})$ points would not be independent, and thus the argument of Corollary~\ref{cor:many} would not apply. A more careful argument would have to be crafted. Of course, one would still need a further idea to then reduce $N$ from $\Theta(n^{1+\gamma})$ down to $n$.

\section*{Acknowledgments}
We thank Rados{\l}aw Adamczak for pointing out how to derive Corollary~\ref{cor:trace} from Theorem~\ref{thm:latala}, and for pointing out the reference \cite{AW13}, which uses a more involved but similar argument.

\bibliographystyle{alpha}
\bibliography{biblio}

\end{document}